\documentclass[preprint]{imsart}

\usepackage{graphicx}

\usepackage{amsmath}
\usepackage{amsthm}
\usepackage{amsfonts}
\usepackage{amssymb}
\usepackage[colorlinks]{hyperref}

\theoremstyle{plain}
\newtheorem{thm}{Theorem}[section]

\newtheorem{prop}[thm]{Proposition}

\theoremstyle{definition}
\newtheorem{example}[thm]{Example}
\newtheorem{defi}[thm]{Definition}

\theoremstyle{remark}
\newtheorem{remark}[thm]{Remark}

\def\1c{\mathbb{I}_C(x)}

\def\X{\mathcal{X}}

\def\n0{n_{0}}

\begin{document}

\begin{frontmatter}

\title{The Containment Condition and AdapFail
  algorithms
}
\runtitle{Containment and AdapFail}

\begin{aug}
\author{\fnms{Krzysztof} \snm{{\L}atuszy\'{n}ski}}%
\ead[label=e1]{latuch@gmail.com}%
\and
\author{\fnms{Jeffrey S.} \snm{Rosenthal}\corref{}
\ead[label=e3]{jeff@math.toronto.edu}
}

\runauthor{K. {\L}atuszy\'nski and J.S. Rosenthal}

\affiliation{University of Warwick}

\address{K. {\L}atuszy\'nski\\ Department of Statistics\\
University of Warwick\\ CV4 7AL, Coventry, UK \\
\printead{e1}
}

\address{J. S. Rosenthal\\Department of Statistics\\
University of Toronto\\ Toronto, Ontario, Canada,
M5S 3G3\\
\printead{e3}}
\end{aug}

\medskip
(June 2013; slightly revised December 2013)

\begin{abstract} This short note investigates convergence of
adaptive MCMC algorithms, i.e.\ algorithms which modify the Markov
chain update probabilities on the fly.  We focus on the Containment condition
introduced in \cite{roberts2007coupling}.
We show that if the Containment condition
is \emph{not} satisfied, then the algorithm will perform very poorly.
Specifically, with positive probability,
the adaptive algorithm will be asymptotically less efficient then
\emph{any} nonadaptive ergodic MCMC algorithm. We call such algorithms
\texttt{AdapFail}, and conclude that they should not be used.
\end{abstract}



\end{frontmatter}

\section{Introduction} \label{sec_intro}

Markov chain Monte Carlo (MCMC) algorithms are used to sample from
complicated probability distributions.  They proceed by simulating
an ergodic Markov chain with transition kernel $P$ and stationary
distribution of interest, say $\pi$. Unlike in the case of iid
Monte Carlo, the MCMC output
\begin{equation}
  \label{eq:sample}
  X_0, X_1, ..., X_n, ...
\end{equation}
is a correlated sample. Nevertheless, if the Markov
chain is ergodic (i.e., converges in distribution to $\pi$), then
the asymptotic validity is retained under appropriate conditions
(see e.g.\ \cite{MT2009, roberts2004general}).  In particular,
for $M$ large enough, the subsampled random variables
\begin{equation}
  \label{eq:subsample}
  X_M, X_{2M}, ..., X_{nM}, ...
\end{equation}
are approximately independent draws from the target distribution
$\pi$.  For the MCMC-based statistical
inference to be reliable, it is
essential to design algorithms that mix quickly , i.e.\
for which the asymptotic iid property in \eqref{eq:subsample} holds
with reasonably small $M$. (Note however that for estimation
purposes, subsampling is desirable only if the cost of using the
sample is substantial compared to the cost of generating samples,
otherwise the entire sample should be used;
see Section 3.6 of \cite{geyer1992}.)

In a typical MCMC setting, the algorithm is determined by a Markov
chain transition kernel $P_{\theta}$, where $\theta \in \Theta$ is
a high dimensional tuning parameter, e.g. the covariance matrix
of a Random Walk Metropolis proposal \cite{roberts1997weak,
roberts2001optimal}, or the vector of Random Scan Gibbs Sampler
selection probabilities \cite{latuszynski2010adaptive}. Usually
the parameter space $\Theta$ is large, and for ``good'' values of
$\theta$, the iterates $P_{\theta}^n$ will converge quickly to $\pi$
as $n$ increases, resulting in small $M$ in \eqref{eq:subsample}.
However, such ``good'' values are often very difficult to find, and
for most values of $\theta$ the iterates $P_{\theta}^n$ will converge
arbitrary slowly.

Since a good $\theta$ is difficult to find
manually, the idea of adaptive MCMC was introduced
\cite{gilks1998adaptive,haario2001adaptive} to enable the algorithm
to learn ``on the fly'', and redesign the transition kernel during
the simulation as more and more information about $\pi$ becomes
available. Thus an adaptive MCMC algorithm would apply the transition
kernel $P_{\theta_n}$ for obtaining $X_n$ from $X_{n-1}$, where
the choice of the tuning parameter $\theta_n$ at the $n^{\rm th}$
iteration is itself a random variable which may depend on the whole
history $X_0, X_1, ..., X_{n-1}$ and on $\theta_{n-1}$.
When using adaptive MCMC, one hopes that the adaptive parameter
$\theta_n$ will settle on ``good'' values, and that the adaptive
algorithm will inherit the corresponding good convergence
properties.

Unfortunately, since adaptive algorithms violate
the Markovian property, they are inherently difficult to analyse
theoretically.  Whereas the interest in adaptive MCMC is fuelled by
some very successful implementations for challenging problems
\cite{roberts2009examples, andrieu2008tutorial, richardson2010bayesian,
MarkJim, giordani2008efficient, solonen2012efficient}, many seemingly
reasonable adaptive MCMC algorithms are provably transient or converge
to a wrong probability distribution \cite{atchade2005adaptive,
bai2011containment, latuszynski2010adaptive, kl_path}.  Thus, the
theoretical foundations of adaptive MCMC are a very important topic
which is still under active development.

One general and relatively simple approach to analysing adaptive
MCMC algorithms was presented in \cite{roberts2007coupling}, which
showed that the two properties of \emph{Diminishing Adaptation}
and \emph{Containment} were sufficient to guarantee that an adaptive
MCMC algorithm would converge asymptotically to the correct target
distribution (at some rate).  While the Diminishing Adaptation property
is fairly standard and can be easily controlled by the user, the
Containment property is more subtle and can be challenging to verify
(see e.g.\ \cite{bai2008containment}).  This leads to the question
of how important or useful the Containment condition actually is,
especially since it is known (see e.g. \cite{fort2011convergence})
that Containment is not a necessary condition for the ergodicity of
an adaptive MCMC algorithm.

The purpose of this short note is to show that if Containment
does not hold, then the adaptive algorithm will perform very
poorly.  Specifically, with positive probability the adaptive
algorithm will be asymptotically less efficient then \emph{any}
nonadaptive MCMC algorithm. Here efficiency is understood as the total
variation distance convergence time; see \cite{mira1999ordering} for a
different concept of efficiency in terms of asymptotic variance in the
central limit theorem. In effect, the approximate iid property
in \eqref{eq:subsample} will be violated for \emph{any} finite $M$.
We call such algorithms \texttt{AdapFail}, and conclude that they
should not be used.  In particular, this argues that the Containment
condition is actually a reasonable condition to impose on adaptive
MCMC algorithms, since without it they will perform so poorly as to
be unusable.

This paper is structured as follows. In Section \ref{sec:AdapFail},
we define and characterise the class of \texttt{AdapFail} algorithms.
In Section \ref{sec:Containment}, we relate the \texttt{AdapFail}
property to the Containment condition. In Section \ref{sec:examples},
we present a very simple example to illustrate our results.

\section{The class of \texttt{AdapFail} algorithms}\label{sec:AdapFail}

We first introduce necessary notation; see e.g.\ \cite{MT2009,
roberts2004general, roberts2007coupling} for more complete development
related to Markov chains and adaptive MCMC. Let $P_{\theta},$
parametrized by  $\theta \in \Theta,$ be a transition kernel of
a Harris ergodic Markov chain on $(\mathcal{X},\mathcal{F})$ with
stationary distribution $\pi$. Thus for all $x \in \mathcal{X}$ and
$\theta \in \Theta$ we have $\lim_{n\to \infty} \|P^n_{\theta}(x,
\cdot) - \pi(\cdot)\| = 0,$ where $\| \nu(\cdot) - \mu(\cdot)\|:=
\sup_{A \in \mathcal{F}}|\nu(A) - \mu(A)|$ is the usual total
variation norm.  We shall also use the ``$\varepsilon$ convergence time
function'' $M_{\varepsilon}: \mathcal{X}\times \Theta \to \mathbb{N}
$ defined as
\begin{equation}
  \label{eq:M_eps}
  M_{\varepsilon}(x, \theta) := \inf\{n \geq 1: \|P^n_{\theta}(x,
\cdot) - \pi(\cdot)\| \leq \varepsilon \}.
\end{equation}
Let $\{(X_n,\theta_n)\}_{n=0}^\infty$
be a corresponding adaptive MCMC algorithm,
where $X_n$ is updated from $X_{n-1}$ using $P_{\theta_n}$ for
\emph{some} $\Theta$-valued
random variable $\theta_n$ (which might depend on the
chain's history and on $\theta_{n-1}$).
For the adaptive algorithm,
denote the marginal distribution at time $n$ by  
\begin{equation}
  \label{eq:distr_adap}
  A^{(n)}((x,\theta), B) := \mathbb{P}(X_n \in B | X_0 = x, \theta_0 = \theta),
\end{equation} 
and say that the algorithm is
\emph{ergodic} for starting values $x$ and $\theta$ if
\begin{eqnarray}
  \label{eq:adap_erg}
\lim_{n \to \infty} \| A^{(n)}((x,\theta), \cdot) - \pi(\cdot)\| & = & 0.
\end{eqnarray}
Similarly let the ``$\varepsilon$ convergence time
function'' for the adaptive case be
\begin{equation}
  \label{eq:MA_eps}
  M^A_{\varepsilon}(x, \theta) := \inf\{n \geq 1: \|A^{(n)}((x,\theta),
\cdot) - \pi(\cdot)\| \leq \varepsilon \}.
\end{equation}
In both cases the function $M_{\varepsilon}(x,\theta)$ has the same
interpretation: it is
the number of iterations that the algorithm must take
to be within $\varepsilon$ of stationarity.

We are now ready to define the class of \texttt{AdapFail} algorithms.

\begin{defi}
  Let $\{(X_n, \theta_n)\}_{n=0}^\infty$ evolve according to the dynamics
  of an adaptive MCMC algorithm $\mathcal{A}$, with starting values
  $X_0 = x^*$ and $\theta_0 = \theta^*$.
  We say that $\mathcal{A} \in \texttt{AdapFail}$ if there is
  $\varepsilon_{AF} >0$ such that
\begin{equation}
\lim_{M\to\infty}
\limsup_{n\to \infty}
\mathbb{P}\big(M^A_{\varepsilon_{AF}}(X_n, \theta_n) > M  \; |\;  X_0 =
x^*, \theta_0 = \theta^*\big) \; =: \;\delta_{AF} \ > \ 0  
\, .
\label{eq:AdapFail_def_1}
\end{equation}
\end{defi}

\begin{remark}
  Intuitively, \eqref{eq:AdapFail_def_1} says that the convergence
times of the adaptive algorithm will be larger than any fixed value
$M$, i.e.\ that the algorithm will converge arbitrarily slowly and
thus perform so poorly as to be unusable.
\end{remark}
\begin{remark}
  In our experience, the inner limit in \eqref{eq:AdapFail_def_1}
  will typically exist, so that $\limsup_{n\to\infty}$ can be replaced
  by $\lim_{n\to\infty}$ there (and similarly in the related expressions
  below).  However, without assuming specific details about the type
  of adaptation used, we are unable to make conclusive statements about
  what conditions guarantee this.
\end{remark}
\begin{remark}
  For the probabilities in \eqref{eq:AdapFail_def_1}
  to make sense, the function
  $M_{\varepsilon}^A$ needs to be measurable. This follows from the
  Appendix of \cite{roberts1997geometric}. 
  Moreover, if the inner limit in \eqref{eq:AdapFail_def_1} is denoted as
  $\delta_{AF}(M)$, then this sequence is positive and non-increasing
  as a function of $M$, and will thus
  converge to $\delta_{AF}$ as $M \to \infty$.
\end{remark}
\begin{remark}
  To obtain the approximate iid property of the $\{X_n\}$ in
  \eqref{eq:subsample}, we want the
  distribution of $X_{(n+1)M}$ conditionally on the value of $X_{nM}$
  to be within $\varepsilon$ of the stationary measure, i.e.
  \begin{equation}
    \label{eq:conditional_stationarity}
    \big\|\mathcal{L}\big(X_{(n+1)M}\;|\;X_{nM}\big) - \pi \big\|
      \, \leq \, \varepsilon. 
  \end{equation}
  Being an \texttt{AdapFail} algorithm means that for any fixed
  $0<\varepsilon \leq \varepsilon_{AF}$ and some fixed
  $\delta_{AF}>0$, we are infinitely often
  in a regime where \eqref{eq:conditional_stationarity}
  is violated for \emph{any} finite $M$, with
  probability at least $\delta_{AF}$, further illustrating its poor
  performance.
\end{remark}

The following two results shed additional light on the
\texttt{AdapFail} class.

\begin{prop}
  Any ergodic nonadaptive MCMC algorithm $P_{\theta}$ is
  not in \texttt{AdapFail}.
\end{prop}
\begin{proof}
  For a nonadaptive chain, the quantity $M_{\varepsilon}^A$ in
  \eqref{eq:AdapFail_def_1} becomes $M_{\varepsilon}$ and $\theta^* =
  \theta$. For arbitrary $\varepsilon > 0$ and 
  $\delta > 0$, we shall show that $\delta_{AF} < 2\delta$,
  from which it follows that $\delta_{AF} = 0$. Indeed,
  first find $n_0$ such that $\|P^{n_0}_{\theta}(x^*, \cdot) - \pi(\cdot)\| <
  \delta$, and then find $M_0$ such that
  $\pi( \{x: M_{\varepsilon}(x, \theta) \leq
  M_0\}) > 1-\delta$.  Then for every $n \geq n_0$ and every $M \geq M_0$,
  we can write
  \begin{eqnarray*}
   \mathbb{P}\big(M_{\varepsilon}(X_n,\theta) > M \;|\; X_0 = x^*\big)
   & \leq & \delta + \pi(\{x: M_{\varepsilon}(x, \theta) > M \}) \;\;
   < \;\; 2\delta.
  \end{eqnarray*}
  The result follows.
\end{proof}

\begin{thm}\label{thm:equivalences}
  For an algorithm $\mathcal{A}$ the following conditions are equivalent.
  \begin{itemize}
  \item[(i)] $\mathcal{A} \in \texttt{AdapFail}$.
  \item[(ii)] there are $\varepsilon > 0$ and $\delta > 0$ such that
    for all $x \in \mathcal{X}$, $\theta \in \Theta$, and $K > 0$,
    \begin{equation*}
      \label{eq:AdapFail_2}
      \limsup_{n \to \infty} \ \mathbb{P}\big( M_{\varepsilon}^A(X_n,
      \theta_n) > KM_{\varepsilon}(x,\theta) \; | \; X_0 = x^*,
      \theta_0 = \theta^* \big) \; \geq \; \delta.
    \end{equation*}
  \item[(iii)] there are $\varepsilon > 0$ and $\delta > 0$ such that
    for all $\theta \in \Theta$, $K > 0$, and $y^* \in \mathcal{X}$,
    \begin{equation*}
      \label{eq:AdapFail_3}
      \limsup_{n \to \infty} \ \mathbb{P}\big( M_{\varepsilon}^A(X_n,
      \theta_n) > KM_{\varepsilon}(Y_n,\theta) \; | \; X_0 = x^*,
      \theta_0 = \theta^*, Y_0 = y^*  \big) \; \geq \; \delta,
    \end{equation*}
    where $\{Y_n\}$ is a Markov chain which follows the dynamics
    $P_{\theta}$ and is independent of the adaptive process $\{X_n\}$.
  \end{itemize}
Moreover, in $(ii)$ and $(iii)$ we can take $\delta = \delta_{AF}$.
\end{thm}

\begin{proof}
 The implication $(i) \Rightarrow (ii)$ with $\delta = \delta_{AF}$
 and $\varepsilon = \varepsilon_{AF}$ is immediate.
 To verify $(ii)
 \Rightarrow (iii),$ fix $\delta^* >0$ and using monotonicity of the
 total variation distance (see \cite{roberts2004general}) take $n_0$  such that
 $\|P^n_{\theta}(y^*, \cdot) - \pi(\cdot)\| \leq \delta^*$ for every
 $n>n_0$. Next, find $M_0$ such that  $\pi(\mathcal{X}_{M_0}) >
 1-\delta^*,$ where $ \mathcal{X}_{M_0} = \{x: M_{\varepsilon_{AF}}(x, \theta) \leq
  M_0\}.$
 Then for fixed $\theta$, $K$, and $y^*$, compute
 \begin{align*}
   \mathbb{P}\big( &M_{\varepsilon_{AF}}^A(X_n,
      \theta_n)  >  KM_{\varepsilon_{AF}}(Y_n,\theta) \; | \; X_0 = x^*,
      \theta_0 = \theta^*, Y_0 = y^*  \big) \\
& = \;\;  \int_{\mathcal{X}} \mathbb{P}\big( M_{\varepsilon_{AF}}^A(X_n,
      \theta_n) > KM_{\varepsilon_{AF}}(x,\theta) \;  | \; X_0 = x^*,
      \theta_0 = \theta^* \big) P^n_{\theta}(y^*, dx) \quad \\
& \geq \;\; \int_{\mathcal{X}} \mathbb{P}\big( M_{\varepsilon_{AF}}^A(X_n,
      \theta_n) > KM_{\varepsilon_{AF}}(x,\theta) \;  | \; X_0 = x^*,
      \theta_0 = \theta^* \big) \pi(dx) \; - \; \delta^* \quad \\
& \geq \;\; \int_{\mathcal{X}_{M_0}} \mathbb{P}\big( M_{\varepsilon_{AF}}^A(X_n,
      \theta_n) > KM_0 \;  | \; X_0 = x^*,
      \theta_0 = \theta^* \big) \pi(dx) \; - \;  \delta^*  \quad \\
& \geq \;\; (1-\delta^*) \mathbb{P}\big( M_{\varepsilon_{AF}}^A(X_n,
      \theta_n) > KM_0 \;  | \; X_0 = x^*,
      \theta_0 = \theta^* \big) \; - \;  \delta^*. 
 \end{align*}
Consequently
\begin{align*}
  \limsup_{n \to \infty} \ & \mathbb{P}\big( M_{\varepsilon_{AF}}^A(X_n,
      \theta_n)  >  KM_{\varepsilon_{AF}}(Y_n,\theta) \; | \; X_0 = x^*,
      \theta_0 = \theta^*, Y_0 = y^*  \big) \\
& \geq \;\; (1-\delta^*) \limsup_{n \to \infty} \ \mathbb{P}\big( M_{\varepsilon_{AF}}^A(X_n,
      \theta_n) > KM_0 \;  | \; X_0 = x^*,
      \theta_0 = \theta^* \big) \; - \;  \delta^*. 
 \end{align*}
Since $\delta^*$ was arbitrary, $(iii)$ follows from $(ii)$ with
$K=KM_0$ and $\delta = \delta_{AF}.$
For $(iii) \Rightarrow (i)$,
notice that $M_{\varepsilon_{AF}}(Y_n,\theta) \geq 1$, so $(iii)$ gives
 \begin{equation*}
   \limsup_{n\to \infty} \
   \mathbb{P}\big( M_{\varepsilon_{AF}}^A(X_n, \theta_n)
   > K \; | \; X_0 = x^*, \theta_0 = \theta^* \big) \; > \; \delta_{AF},
   \quad \textrm{for every} \ K >0.
 \end{equation*}
The result follows by taking $K \to \infty$.
\end{proof}

\begin{remark}
  Condition $(iii)$ has the interpretation that if we run the adaptive
  algorithm $\{X_n\}$ and a nonadaptive $\{Y_n\}$ independently on two
  computers next to each other, and monitor the $\varepsilon$
  convergence time of both algorithms, then as the simulation
  progress, the $\varepsilon$ convergence time of
  the adaptive algorithm will infinitely often
  be bigger by an arbitrarily large factor $K$,  with probability at least $\delta$,
  i.e.\ $\{X_n\}$ will be arbitrarily worse than $\{Y_n\}$
  (no matter how bad are the tuning parameters $\theta$
  and starting point $Y_0$ for $\{Y_n\}$).
\end{remark}

\section{Relation to the Containment condition}\label{sec:Containment}

The following condition was introduced in
\cite{roberts2007coupling} as a tool to analyse adaptive MCMC algorithms:
\begin{defi}[Containment Condition]
  The algorithm  $\mathcal{A}$ with starting values
  $X_0=x^*$ and $\theta_0 = \theta^*$ satisfies Containment, if for
  all $\varepsilon > 0$ the
  sequence $\{M_{\varepsilon}(X_n, \theta_n)\}_{n=0}^{\infty}$ is
  bounded in probability.
\end{defi}
It is augmented by the usual requirement of Diminishing Adaptation:
\begin{defi}[Diminishing Adaptation] The algorithm  $\mathcal{A}$ with starting values
  $X_0=x^*$ and $\theta_0 = \theta^*$ satisfies Diminishing
  Adaptation, if
  \begin{equation*}
    \label{eq:dim_adap}
    \lim_{n\to \infty} D_n = 0 \quad \textrm{in probability, where}
      \quad D_n:= \sup_{x \in
      \mathcal{X}}\|P_{\theta_{n+1}}(x, \cdot) -  P_{\theta_{n}}(x, \cdot)\|.
  \end{equation*}
\end{defi}

Containment has been extensively studied in \cite{roberts2007coupling}
and \cite{bai2011containment} and verified for large classes of
adaptive MCMC samplers (c.f.\ also \cite{roberts2009examples,
latuszynski2010adaptive}). Together with Diminishing Adaptation, it
guarantees ergodicity. As illustrated in the next section,
it is \emph{not} a necessary condition.  However, it still turns out to be
an appropriate condition to require, due to the following result.

\begin{thm} \label{thm:AdapFail_equiv_noCont}
  Assume the Diminishing Adaptation is satisfied.  Then the Containment
  condition does not hold for $\mathcal{A}$ if and only if
  $\mathcal{A} \in \texttt{AdapFail}$.
\end{thm}

\begin{proof} The proof utilises a construction similar to the coupling
  proof of Theorem~1 of \cite{roberts2007coupling} (see also
  \cite{roberts2013note}).  First, by the
  Diminishing Adaptation property, for any fixed $\delta_c>0$,
  $\varepsilon_c>0$, and integer $M \geq 1$, we can choose $n$ big enough
  that
\begin{equation}\label{eq:bad_set}
\mathbb{P}\Big( 
\bigcup_{k=1}^{M} \{D_{n+k} > {\varepsilon_c \over 2M^2} \}\Big) \; \leq
\; {\delta_c \over 2}.\end{equation}
Now, on the set $\bigcap_{k=1}^M \{ D_{n+k} \leq {\varepsilon_c \over 2 M^2}
\} $ for transition kernels $P_{\theta_n},
P_{\theta_{n+1}}, ...,  P_{\theta_{n+M}},$ by the
triangle inequality we have \begin{align} \label{eq:diminish_telescoping}
\sup_{x \in
  \mathcal{X}}\Big\| & \big(\prod_{k=0}^{M} P_{\theta_{n+k}}\big)(x, \cdot) -
  P^M_{\theta_{n}}(x, \cdot)\Big\| \; \leq \; \\ & \leq \; \sum_{k=1}^{M} \sup_{x \in
  \mathcal{X}} \bigg\|  \Big(\big(\prod_{i=0}^{k+1} P_{\theta_{n+i}}\big)P^{M-k-1}_{\theta_n}\Big)(x, \cdot) -
 \Big(\big(\prod_{i=0}^{k}
 P_{\theta_{n+i}}\big)P^{M-k}_{\theta_n}\Big)(x, \cdot) 
  \bigg\| \qquad  \nonumber \\ \nonumber & \leq \; \sum_{k=1}^{M}
  (k+1){\varepsilon_c \over 2M^2} \; = \; {M+1 \over 4M} \varepsilon_c
  \; < \; {\varepsilon_c \over 2}. 
\end{align}
Consequently we conclude that for $n$ large enough,
 \begin{equation}
  \label{eq:diminishing_for_sequence}
  \mathbb{P}\Big(\textrm{LHS of \eqref{eq:diminish_telescoping} } <
  {\varepsilon_c \over 2} \Big) \; > \; 1-{\delta_c \over 2}.
\end{equation}
  For the ``only if'' part of the theorem,
  note that if Containment does not hold, then for
  the adaptive algorithm in question,
  there is $\varepsilon_c > 0$ and $\delta_c > 0$ such that
  \begin{eqnarray}
    \label{eq:no_containment} \quad
    \forall {M, \,n_0}, \ \exists {n >n_0} \quad \textrm{s.t.} \quad
    \mathbb{P} (M_{\varepsilon_c}(X_n, \theta_n) > M) > \delta_c .
  \end{eqnarray}
By \eqref{eq:diminishing_for_sequence}, we obtain
 \begin{eqnarray}
    \label{eq:AdapFail_obtained} \quad
    \forall {M, \,n_0}, \ \exists {n >n_0} \quad \textrm{s.t.} \quad
    \mathbb{P} (M^{A}_{\varepsilon_c / 2}(X_n, \theta_n) > M) >
    {\delta_c \over 2}.
  \end{eqnarray}
which implies the
\texttt{AdapFail} condition with $\varepsilon_{AF} \geq \varepsilon_c/2$
and  $\delta_{AF} \geq \delta_c/2$.

The proof for the ``if'' part of the theorem is essentially the same.
From~\eqref{eq:AdapFail_def_1}
and~\eqref{eq:diminishing_for_sequence}, one obtains
\eqref{eq:no_containment} with $\varepsilon_c \geq \varepsilon_{AF} /2$ and
$\delta_c \geq \delta_{AF}/2$.
\end{proof}

\begin{remark}
Without assuming Diminishing Adaptation,
Theorem~\ref{thm:AdapFail_equiv_noCont} does not hold.
For example, if $P$ is a fixed ergodic Markov chain,
and $I$ is the identity kernel (which does not move at all), then the
adaptive scheme which simply alternates between $P$ and $I$
converges well (at half-speed compared to $P$) and is not in AdapFail.
However, this scheme violates Containment, since if
$\theta_1$ is the adaptive parameter corresponding to $I$, then 
$M_{\varepsilon}(x, \theta_1) = \infty$.
\end{remark}

\section{A very simple example}\label{sec:examples}

In this section, we analyse a very simple example of an adaptive algorithm,
to illustrate our results about \texttt{AdapFail}.

\begin{example}
Consider the toy example from \cite{fort2011convergence}
with state space $\X=\{0,1\}$ and stationary distribution
$\pi=(1/2,1/2)$, with Markov transition kernels
$$
P_\theta \ = \
\left(
\begin{matrix}
1-\theta & \theta \cr \theta & 1-\theta
\end{matrix}
\right)
\, .
$$
Suppose the $n^{\rm th}$ iteration of the Markov chain
uses kernel $P_{\theta_n}$ (independent of the chain's past
history), where $\theta_n > 0$ and $\sum_n \theta_n = \infty$
but $\theta_n \to 0$ (e.g.\ $\theta_n=1/n$).  Since
the $\theta_n$ converges, clearly Diminishing Adaptation is
satisfied.  On the other hand, as
$\theta \to 0$, $M_\epsilon(x,\theta) \to \infty$.  Hence, this
adaptive algorithm does {\it not} satisfy Containment.  So, by
the above theorems, this algorithm converges more slowly than any
fixed non-adaptive algorithm.
But since $\sum_n \theta_n = \infty$, this algorithm is still
ergodic \cite{fort2011convergence}.  We thus have a (very simple)
example of an adaptive algorithm which is ergodic, but is nevertheless
in AdapFail and has very poor convergence properties.
(A similar result presents if instead $\theta_n \to 1$ with
$\sum_n (1-\theta_n) = \infty$.)
\end{example}

\section{Acknowledgements} K{\L} acknowledges funding from CRISM and
other grants from EPSRC. JSR acknowledges funding from
NSERC of Canada. We thank Gersende Fort and
Gareth O.\ Roberts for helpful discussions.

\bibliographystyle{alpha} 
\bibliography{../../../references}

\end{document}